\documentclass[12pt]{iopart}

\usepackage{enumerate}
\usepackage{xcolor}
\usepackage{mathrsfs}
\usepackage{dsfont}
\usepackage{graphicx}
\usepackage{hyperref}
\usepackage{amssymb}
\usepackage{amsthm}
\usepackage{array}
\newcolumntype{L}{>{$}l<{$}} 

\expandafter\let\csname equation*\endcsname\relax
\expandafter\let\csname endequation*\endcsname\relax
\usepackage{amsmath}
\usepackage{cases} 

\newcommand{\eps}{\varepsilon}
\newcommand{\RR}{\mathds R}
\newcommand{\EE}{\mathds E}

\newtheorem{lemma}{Lemma}
\theoremstyle{remark}
\newtheorem{remark}{Remark}

\begin{document}

\title[]{Mean First Passage Times and Eyring-Kramers formula for Fluctuating Hydrodynamics}

\author{Jingbang Liu}
\address{Mathematics Institute, University of Warwick, Coventry
  CV4 7AL, United Kingdom}
\address{Department of Mathematics, University of Oslo, Oslo, 0851, Norway}
\eads{\mailto{jingbanl@.uio.no}}

\author{James E. Sprittles}
\address{Mathematics Institute, University of Warwick, Coventry
  CV4 7AL, United Kingdom}
\eads{\mailto{J.E.Sprittles@Warwick.ac.uk}}

\author{Tobias Grafke}
\address{Mathematics Institute, University of Warwick, Coventry
  CV4 7AL, United Kingdom}
\eads{\mailto{T.Grafke@warwick.ac.uk}}

\date{\today}

\begin{abstract}
  Thermally activated phenomena in physics and chemistry, such as
  conformational changes in biomolecules, liquid film rupture, or
  ferromagnetic field reversal, are often associated with
  exponentially long transition times described by Arrhenius' law. The
  associated subexponential prefactor, given by the Eyring-Kramers
  formula, has recently been rigorously derived for systems in
  detailed balance, resulting in a sharp limiting estimate for
  transition times and reaction rates. Unfortunately, this formula
  does not trivially apply to systems with conserved quantities, which
  are ubiquitous in the sciences: The associated zeromodes lead to
  divergences in the prefactor. We demonstrate how a generalised
  formula can be derived, and show its applicability to a wide range
  of systems, including stochastic partial differential equations from
  fluctuating hydrodynamics, with applications in rupture of nanofilm
  coatings and social segregation in socioeconomics.
\end{abstract}

\maketitle

\section{Introduction}


Metastability is a well-known phenomenon appearing in many areas of
natural sciences: A stochastic system spends a long time near some
typical configuration, but can rarely switch (transition) to a drastically
different configuration. The corresponding waiting times are known to
be exponentially large in the noise strength.

The typical picture is that of a stochastic diffusion in a potential
landscape, where local minima correspond to long-lived
states. Fluctuations can then push the system across a potential
barrier into another local minimum, in which it will remain for long
times. While chemical reactions, conformational changes in
biomolecules, protein folding, or magnetic field reversal in
ferromagnets are well-known examples of this phenomenon, similar ideas
can be found in rather broad areas of science, such as in brain
activity~\cite{brinkman-yan-maffei-etal:2022}, quenched disorder in
semiconductors~\cite{garcia-hofmann:2024}, tipping points in Earth's
climate~\cite{ashwin-heydt:2020} including warm-water currents in the
north Atlantic~\cite{lohmann-dijkstra-jochum-etal:2024}, or ecosystem
collapse~\cite{bashkirtseva-ryashko:2011}.

Concretely, take the gradient diffusion for $X_t\in\RR^n$,
\begin{equation}
  \label{eq:gradient-flow}
  dX_t = -\nabla U(X_t)\,dt + \sqrt{2\eps}\,dW_t\,,
\end{equation}
where $U:\RR^n\to \RR$ is the potential, and $W_t$ white-in-time
Brownian motion. For this system and in the limit of small noise,
$\eps\to0$, transitions between two local minima $x_-$ and $x_+$ of
$U(x)$, under some assumptions on the
potential~\cite{bovier-eckhoff-gayrard-etal:2004}, must happen through
a saddle point $x_s$, and the expected transition time $\tau$ is given
by the \emph{Eyring-Kramers} formula
\begin{equation}
  \label{eq:eyring-kramers-nomobility}
  \tau = \frac{2\pi}{|\lambda_-|}\sqrt{\frac{|\det H_s|}{\det H_-}} e^{\Delta U/\eps}
\end{equation}
asymptotically sharp in the limit of vanishing $\eps$. Here,
$H_-=\nabla\nabla U(x_-)$ and $H_s=\nabla\nabla U(x_s)$ are the
Hessian of the potential $U$ at the starting fixed point $x_-$ and at
the saddle $x_s$ , respectively, and $\lambda_-$ is the single
negative eigenvalue of $H_s$, corresponding to the single unstable
direction of the saddle point. The exponential scaling with the energy
barrier height $\Delta U = U(x_s)-U(x_-)$ is known as \emph{Arrhenius'
  law}~\cite{arrhenius:1889}, and can be made rigorous within sample
path large deviation theory as established
in~\cite{freidlin-wentzell:2012} in more general cases than just
gradient diffusions. The pre-exponential factor is also known for
almost a century~\cite{eyring:1935, kramers:1940}, but has only
recently been proven
rigorously~\cite{bovier-eckhoff-gayrard-etal:2004, berglund:2013,
  bouchet-reygner:2016, landim-seo:2018,
  landim-mariani-seo:2019}. Within these works, it is possible to
consider the more general case of a diffusion in a potential landscape
with \emph{mobility} $M(x):\RR^n\to \RR^{n\times n}$, which is
positive definite and symmetric, via
\begin{equation}
  \label{eq:gradient-flow-mobility}
  dX_t = -M(X_t)\nabla U(X_t)\,dt + \eps \nabla\cdot M(X_t)\,dt + \sqrt{2\eps} M_{1/2}(X_t)\,dW_t\,,
\end{equation}
where $M_{1/2}(x):\RR^n\to\RR^{n\times n}$ is the unique positive
definite matrix for which $M_{1/2} M_{1/2}^T = M$ and $(\nabla\cdot
M)_j (x) = \sum_{i} \partial_{x_i} M_{ij}(x)$. In this case, the
Eyring-Kramers formula reads~\cite{landim-mariani-seo:2019}
\begin{equation}
  \label{eq:eyring-kramers-theirs}
  \tau = \frac{2\pi}{\mu_-} \sqrt{\frac{|\det H_s|}{\det H_-}} e^{\Delta U/\eps}\,,
\end{equation}
where $\mu_-$ is the unique negative eigenvalue of the matrix $M(x_s)
H_s$.

While the generalised gradient diffusion with
mobility~(\ref{eq:gradient-flow-mobility}) can still be interpreted as
system that minimises the potential $U$, just with a position
dependent metric given by the mobility, it opens up a much wider class
of physical phenomena beyond the overdamped Langevin
equation~(\ref{eq:gradient-flow}). In particular, if further
generalizing to the functional setup and allowing generalised gradient
diffusions in function spaces or spaces of probability measures, it
includes hydrodynamic limits of interacting particle systems, lattice
gases, pedestrian dynamics, traffic flow, etc, all of which can be
seen as (functional) gradient flows of some entropy functional for a
(generalised) Wasserstein
metric~\cite{jordan-kinderlehrer-otto:2006}. For example, the large
number of particles limit of many non-interacting random
walkers is given by the stochastic diffusion equation for a density
$\rho(x,t)$,
\begin{equation*}
  \partial_t \rho = \Delta \rho + \sqrt{2\eps}\nabla\cdot(\sqrt{\rho}\eta)\,,
\end{equation*}
with $\eta$ spatio-temporal white noise, and $\eps=1/N$ for $N$
random walkers. It can be interpreted as a functional gradient flow
\begin{equation}
  \label{eq:functional-gradient-flow}
  \partial_t \rho = -M(\rho) \frac{\delta E[\rho]}{\delta \rho} + \sqrt{2\eps}M_{1/2}(\rho)\eta
\end{equation}
in the entropy landscape
\begin{equation*}
  E[\rho] = \int \rho\log \rho\,dx
\end{equation*}
and with mobility operator
\begin{equation}
  \label{eq:IPS-mobility}
  M(\rho)\xi = \nabla\cdot(\rho\nabla \xi)
\end{equation}
(and thus $M_{1/2}(\rho)\xi = \nabla\cdot(\sqrt{\rho}\xi)$). The main
point of this paper is to generalise the Eyring-Kramers
formula~(\ref{eq:eyring-kramers-theirs}) to generalised gradient
systems~(\ref{eq:functional-gradient-flow}) with conserved quantities.

\subsection{Main Result}
\label{sec:main-result}

The major problem in applying the Eyring-Kramers
formula~(\ref{eq:eyring-kramers-theirs}) to generalised gradient
systems of the form~(\ref{eq:functional-gradient-flow}) is that in
almost all cases of physical relevance, the mobility operator is not
positive definite, but instead features zero-eigenvalues corresponding
to conserved quantities. For example the mobility
in~(\ref{eq:IPS-mobility}) has a zero eigenvalue, with constant
functions being the corresponding eigenfunctions, that is associated
with conservation of particle number for the underlying particle
diffusion equation. This situation is generic in hydrodynamic limits,
which often conserve mass, momentum, energy, etc. Our main result is a
modification to the Eyring-Kramers formula, which corrects for the
conserved quantity. For a single conserved quantity, it is given by
\begin{equation}
  \label{eq:eyring-kramers-ours}
  \tau = \frac{2\pi}{\mu_-} \sqrt{\frac{|\det H_s|}{\det H_-}} \sqrt{\frac{\hat m\cdot H_s^{-1}\hat m}{ \hat m\cdot H_-^{-1}\hat m}} e^{\Delta U/\eps}
\end{equation}
where $\hat m$ is the vector normal to the conserved quantity submanifold,
and where stable fixed point $x_-$ and saddle point $x_s$ have to be
appropriately re-interpreted. An extension to multiple conserved quantities will also be provided in section~\ref{sec:cons-quant-mobil}.

We remark that while a non-invertible mobility operator leads to
divergences in the naive formula for the prefactor, a similar
situation may also occur on the level of a diffusion without mobility,
where the Hessian of the potential itself might be degenerate at the
saddle point or at the fixed point, as discussed for example
in~\cite{berglund-gentz:2010, berglund:2013}.

In the following, we will derive
equation~(\ref{eq:eyring-kramers-ours}) via a formal asymptotic
expansion. In particular, we will compute the asymptotics of the mean
first passage time in section~\ref{sec:mean-first-passage} through a
boundary layer analysis and Laplace asymptotics, incorporating the
complications of the conserved quantity. We will then demonstrate the
applicability of the formula by computing mean first passage times for
a simple toy model in section~\ref{sec:two-dimens-grad}, and to two
more realistic stochastic partial differential equations describing
liquid thin film rupture in section~\ref{sec:stoch-hydr-thin}, and
urban segregation in a socioeconomic model of social dynamics in
section~\ref{sec:soci-dynam-urban}.

\section{Mean First Passage Time and Laplace Asymptotics}
\label{sec:mean-first-passage}

In this section, we will derive
equation~(\ref{eq:eyring-kramers-ours}) via a formal asymptotic
expansion. We start with deriving a generic formula for a general
stochastic differential equation and perform its boundary layer
analysis in section~\ref{sec:asympt-expans-bound}. In
section~\ref{sec:laplace-asymptotics}, we then specialize to the case
of (non-degenerate) gradient flows with mobility~(\ref{eq:gradient-flow-mobility}) and apply Laplace asymptotics to
derive the well-known case of the literature. Lastly, in
section~\ref{sec:cons-quant-mobil}, we apply the same reasoning to the
case at hand, namely gradient flows with degenerate mobility, where a
conserved quantity yields a mobility matrix that is no longer positive
definite. This chain of arguments yields our newly proposed formula,
which we subsequently discuss in the context of stochastic
hydrodynamics in section~\ref{sec:funct-grad-flows}. While the
original chain of arguments could be phrased in the more rigorous
language of capacity theory as well, it is at this functional stage
that a rigorous proof is much harder to achieve, and we thus resort to
formal arguments throughout. We present our final full computational
scheme in section~\ref{sec:full-comp-scheme} that forms the basis for
our examples in section~\ref{sec:two-dimens-grad}
to~\ref{sec:soci-dynam-urban}.

Consider first the general stochastic differential equation for
$X_t\in\RR^n$,
\begin{equation}
  \label{eq:SDE-generic}
  dX_t = b(X_t)\,dt + \sqrt{\eps}\sigma(X_t)\,dW_t\,,
\end{equation}
where $b:\RR^n\to\RR^n$ is the deterministic drift,
$\sigma:\RR^n\to\RR^{n\times n}$ defines the noise covariance matrix
$a(x) = \sigma(x)\sigma^T(x)$, and $W_t$ is $n$-dimensional Brownian
motion. We assume the case where there is a stable fixed point
$x_-\in\RR^n$ such that $b(x_-)=0$ and the eigenvalues of $\nabla b(x_-)$
all have negative real part. We are interested in the time it takes
the process to first exit the basin of attraction $B$ of $x_-$ starting at
$x\in B$,
\begin{equation*}
  T_B(x) = \inf\{t>0\ |\ X_t\notin B\}\,.
\end{equation*}
$T_B(x)$ is a random variable, and its expectation $w_B(x)=\EE
T_B(x)$, the so-called \emph{mean first passage time}, fulfills the
inhomogeneous stationary Kolmogorov equation~\cite{gardiner:2009}
\begin{equation}
  \label{eq:BKE}
  \begin{cases}
    \mathcal L w_B(x) = -1 & \text{for } x\in B\\
    w_B(x) = 0 & \text{for } x\in\partial B\,,
  \end{cases}
\end{equation}
where $\partial B$ is the boundary of the basin of attraction of $x_-$,
for which $\hat n\cdot b(u)=0 \ \forall \ u\in\partial B$, with $\hat
n$ being the outwards pointing normal vector to $\partial B$. Here,
\begin{equation}
  \label{eq:generator}
  \mathcal L = b(x)\cdot \nabla + \tfrac12 \eps a(x):\nabla\nabla
\end{equation}
is the generator of the SDE~(\ref{eq:SDE-generic}), from which we can
deduce the invariant distribution $\rho_\infty(x)$ through the stationary
Fokker-Planck equation
\begin{equation*}
  \mathcal L^\dagger \rho_\infty = 0\,,
\end{equation*}
where $\mathcal L^\dagger$ is the $L^2$-adjoint of the generator.

In the case of the gradient flow~(\ref{eq:gradient-flow-mobility}),
the invariant distribution is given as Gibbs distribution through the
potential itself,
\begin{equation}
  \label{eq:gibbs}
  \rho_\infty(x) = C e^{-U(x)/\eps}\,.
\end{equation}
Further, there is a distinguished point $x_s$ on $\partial B$ for
transitions from \textcolor{red}{$x_-$} out of $B$, given through the \emph{barrier
  height}~\cite{berglund:2013},
\begin{equation*}
  \Delta U = \inf_{u\in\partial B}\left( U(u)-U(x_-)\right)
\end{equation*}
i.e.~the smallest potential barrier encountered by continuous curves
starting at $x_-$ and leaving through $\partial B$. The point at which
this barrier is taken, $x_s\in \RR^n$, is assumed to be a saddle point
with a single unstable direction, i.e.~$\nabla U(x_s)=0$ and $M(x_s)
\nabla\nabla U(x_s)$ having exactly one negative eigenvalue $\mu_-$,
and $n-1$ positive eigenvalues. In general there might be multiple
saddles, all of which are dominated by $x_s$, which is therefore
called the \emph{relevant saddle}. In the following, we write
$M_s=M(x_s)$.

From large deviation theory~\cite{freidlin-wentzell:2012} it is then
known that
\begin{equation}
  \label{eq:kramers}
  w_b(x_-) \asymp e^{\eps^{-1}(U(x_s)-U(x_-))}\,,
\end{equation}
which determines the exponential part of the mean first passage
time, recovering \emph{Arrhenius' law}~\cite{arrhenius:1889}. The
purpose of the Eyring-Kramers law, and the goal of this paper, is to
go beyond this mere exponential scaling law, and get sharp asymptotics
of the prefactor omitted in~(\ref{eq:kramers}).

\subsection{Asymptotic expansion and boundary layer analysis}
\label{sec:asympt-expans-bound}

\begin{figure}
  \begin{center}
    \includegraphics[width=80pt]{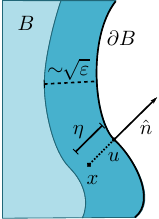}
  \end{center}
  \caption{Schematic depiction of the boundary layer along $\partial
    B$: In an $\mathcal O(\sqrt{\eps})$ vicinity, a point $x\in B$ is
    expanded along normal direction $\hat n$ from boundary point
    $u\in\partial B$, with local coordinate $\eta>0$.}
  \label{fig:boundary-layer}
\end{figure}

In order to get access to the prefactor, and loosely
following~\cite{gardiner:2009}, we assume $w_B(x)\asymp e^{K/\eps}$ as
estimated by large deviation theory~(\ref{eq:kramers}) and work with
\begin{equation*}
  \tau(x) = e^{-K/\eps} w_B(x)\,,
\end{equation*}
which fulfills, via~(\ref{eq:BKE}), the Kolmogorov equation
\begin{equation*}
  \begin{cases}
    \mathcal L \tau(x) = -e^{-K/\eps} & \text{for } x\in B\\
    \tau(x) = 0 & \text{for } x\in\partial B\,,
  \end{cases}
\end{equation*}
so that for $\eps\to0$ the right hand side vanishes and the Kolmogorov
equation becomes homogeneous. Since the diffusive term in the
generator is $\mathcal O(\eps)$ as well, we can assume that within $B$
the variable $\tau(x)$ is merely advected and thus constant,
$\tau(x)=C_0$, to leading order in $\eps$. We need to consider only
the behavior in a small $\mathcal O(\sqrt{\eps})$ boundary layer near
$\partial B$.

For a point $x$ near the boundary $\partial B$, we choose coordinates
\begin{equation*}
  x = u - \sqrt{\eps} \eta \hat n\,,
\end{equation*}
for $\eta>0$ and $u\in\partial B$, compare
figure~\ref{fig:boundary-layer}. In this boundary layer, to leading
order, we therefore have
\begin{equation*}
  \mathcal L \approx (b(x)\cdot\hat n)\frac1{\sqrt{\eps}}\partial_\eta + \underbrace{\hat n\cdot a(x)\hat n}_{\alpha(u)} \partial_\eta^2\,.
\end{equation*}
Since $x$ is $\mathcal O(\sqrt{\eps})$-close to $u\in\partial B$, we can expand
\begin{equation*}
  b(x)\cdot \hat n = \underbrace{b(u)\cdot \hat n}_{=0} + \hat n \cdot \nabla b(u)(x-u) + \mathcal O(|x-u|^2)\,,
\end{equation*}
and introduce the additional quantity $\beta(u)$ through
\begin{equation*}
  \hat n \cdot\nabla b(u) (x-u) = -\sqrt{\eps}\eta \underbrace{\hat n\cdot \nabla b(u)\hat n}_{\beta(u)}\,.
\end{equation*}
We will later see that $\beta(u)$ at the saddle $u=x_s$ is related to
the unstable eigenvalue $\mu_-$. Now, all terms are $\mathcal
O(\eps^0)$ and we arrive at
\begin{equation}
  \label{eq:wkb-eta}
  0=\mathcal L\tau(\eta) = \eta \beta(u)\partial_\eta \tau(\eta) + \alpha(u)\partial_\eta^2\tau(\eta)\,,
\end{equation}
with $\alpha, \beta$ being the leading-order contributions of the
normal terms of $a$ and $b$, respectively. Equation~(\ref{eq:wkb-eta})
is solved by
\begin{equation*}
  \tau(\eta) = C_1(u) \int_0^\eta e^{-\frac{\beta(u)}{2\alpha(u)} \eta^2}\,d\eta\,.
\end{equation*}
Since we know the limit $\tau(\eta)\xrightarrow{\eta\gg1}C_0$, we must
have the matching condition
\begin{equation*}
  C_1(u) = C_0 \sqrt{\frac{2\beta(u)}{\pi\alpha(u)}}\,,
\end{equation*}
from which we can obtain $C_0$. The connection between the bulk and
the boundary can be exploited when integrating $\mathcal
L\tau(x)=-e^{-K/\eps}$ against the invariant density $\rho_\infty(x)$,
and integrating by parts,
\begin{align*}
  -e^{-K/\eps} \int_B \rho_\infty(x)\,dx &= \int_B\rho_\infty(x) \mathcal L \tau(x)\,dx\\
  &= \int_B\underbrace{(\mathcal L^\dagger \rho_\infty)}_{=0}\tau(x)\,dx + \\
  &\ \ \int_{\partial B}\left(\rho_\infty (\hat n\cdot b(u))\underbrace{\tau(u)}_{\makebox[0pt]{\footnotesize{$\tau=0\,\text{on}\,\partial B$}}} + \eps\big(\rho_\infty \hat n\cdot a(u)\nabla \tau(u) - \underbrace{\tau(u)}_{\makebox[0pt]{\footnotesize{$\tau=0\,\text{on}\,\partial B$}}}\hat n\cdot a(u)\nabla \rho_\infty\big)\right)\,du\\
  &= -\sqrt{\eps}\int_{\partial B} \rho_\infty(u) \alpha(u) \partial_\eta \tau\,du\\
  &= -\sqrt{\frac{2\eps}{\pi}} C_0 \int_{\partial B} \rho_\infty(u)\sqrt{\alpha(u)\beta(u)}\,du
\end{align*}
so that
\begin{equation*}
  C_0 = e^{-K/\eps} \sqrt{\frac{\pi}{2\eps}} \frac{\int_B \rho_\infty(x)\,dx}{\int_{\partial B} \rho_\infty(u)\sqrt{\alpha(u)\beta(u)}\,du}\,.
\end{equation*}
We conclude that in the interior,
\begin{equation}
  \label{eq:wb-general}
  w_B(x_-) = \sqrt{\frac{\pi}{2\eps}} \frac{\int_B \rho_\infty(x)\,dx}{\int_{\partial B} \rho_\infty(u)\sqrt{\alpha(u)\beta(u)}\,du}\,.
\end{equation}

\subsection{Laplace asymptotics}
\label{sec:laplace-asymptotics}

Note that so far we have not made use of the fact that our system is a gradient flow with non-degenerate mobility~(\ref{eq:gradient-flow-mobility}), and that result~(\ref{eq:wb-general}) is asymptotically correct for $\eps\ll 1$ for arbitrary systems. We can now make use of our
explicit knowledge of the invariant measure~(\ref{eq:gibbs}) to apply
Laplace asymptotics to the volume and boundary integrals
in~(\ref{eq:wb-general}). Concretely, that means that in the
numerator, we can approximate
\begin{equation*}
  \int_B\rho_\infty(x)\,dx \approx \frac{(2\pi\eps)^{n/2}}{\sqrt{\det H_-}} e^{-U(x_-)/\eps}\,,
\end{equation*}
since $x_-$ is the minimum of $U$ within $B$, while in the denominator
\begin{equation*}
  \int_{\partial B} \rho_\infty(u) \sqrt{\alpha(u)\beta(u)}\,du \approx \frac{(2\pi\eps)^{(n-1)/2}}{\sqrt{|\det H_s|}}\sqrt{\alpha(s)\beta(s)} (\hat n\cdot H_s^{-1}\hat n)^{-1/2} e^{-U(x_s)/\eps}\,,
\end{equation*}
where the $(\hat n\cdot H_s^{-1}\hat n)$-term comes from the fact that
we integrate the Gaussian integral only over the tangent space to the
separatrix at the saddle, $T_{x_s}\partial B$, as derived in the
appendix in lemma~\ref{lm:detp}. In total, this yields
\begin{equation*}
  w_B(x_-) = \pi \sqrt{\frac{\hat n\cdot H_s^{-1}\hat n}{\alpha(x_s)\beta(x_s)}} \sqrt{\frac{|\det H_s|}{\det H_-}} e^{\Delta U/\eps},
\end{equation*}
where we recall that $\alpha(x_s) = \hat n \cdot M_s\hat n$ and $\beta(x_s) =\hat n\cdot M_s H_s\hat n$, and where the $\mathcal O(\eps)$-part of the drift term $b(x) = -M(x)\nabla U(x) + \eps
\nabla\cdot M(x)$ is subdominant and thus dropped. Here we write $M(x_s)=M_s$, $\nabla\nabla U(x_-)=H_-$ and $\nabla\nabla U(x_s)=H_s$.

Using lemma~\ref{lm:beta} and~\ref{lm:alpha} of the appendix, we
recognise that $\alpha(x_s)$ and $\beta(x_s)$ are connected to $\mu_-$
via
\begin{equation*}
  \beta(x_s) = \frac{\alpha(x_s)}{\hat n\cdot H_s^{-1}\hat n} = \mu_-\,,
\end{equation*}
where $\mu_-$ is the unique negative eigenvalue of $M_s H_s$. We
arrive at the final result
\begin{equation*}
  w_B(x_-) = \frac{\pi}{\mu_-}\sqrt{\frac{|\det H_s|}{\det H_-}} e^{\Delta U/\eps}\,.
\end{equation*}
This demonstrates the capacity theory result from the
literature~\cite{landim-seo:2018} for the case of gradient flows with
position-dependent mobility.

\subsection{Conserved quantities of the mobility matrix}
\label{sec:cons-quant-mobil}

We now consider the case where the system has a conserved quantity,
understood in the sense that the mobility matrix $M(x):\RR^n \to
\RR^{n\times n}$ is no longer positive definite, but positive
semi-definite. In other words, for each $x\in\RR^n$, there exists a
number of zero eigenvalues of $M(x)$, and $M(x)$ is no longer full
rank. As a consequence, since the mobility acts in front of both the
deterministic drift and the stochastic noise, the degrees of freedom
associated with the zero eigenvalues are never changed, and remain a
constant of integration. The concrete value of the conserved quantity
and their nature depends on both the mobility matrix and the initial
condition of the system.

\begin{figure}
  \begin{center}
    \includegraphics[width=0.5\textwidth]{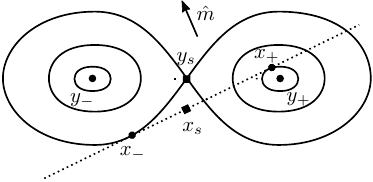}
  \end{center}
  \caption{Conserved quantities of the stochastic process: If $M(x)$
    has a zero eigenvalue with eigenvector $\hat m$, then the gradient
    diffusion~(\ref{eq:gradient-flow-mobility}) will remain
    constrained to a submanifold $S\in\RR^n$ (dotted line) with normal
    vector $\hat m$. Instead of the actual fixpoints
    $\{y_-,y_s,y_+\}$, the relevant points are now the corresponding
    fixed points $\{x_-,x_s,x_+\}$ of the dynamics constrained to the
    submanifold $S$.\label{fig:conservedquantities}}
\end{figure}

For simplicity, we consider a single conserved quantity, so that
$M(x)$ has a unique zero eigenvalue for all $x\in\RR^n$ with
normalised eigenvector $\hat m(x)$, while all its other eigenvalues
are strictly positive. The process~(\ref{eq:gradient-flow-mobility})
then remains constrained to an $(n-1)$-dimensional sub-manifold
$S\subset\RR^n$, with normal vector field $\hat m(x)$, since neither
the gradient drift nor the stochastic force can ever have a
contribution in the direction of $\hat m$.

This situation is depicted in
figure~\ref{fig:conservedquantities}. Note that, while the original
fixed points (both stable, $y_\pm$, and saddle $y_s$) of the system
still exist, they do not lie within $S$. Instead, the effective stable
points $x_+$ and $x_-$ that should be considered for the Laplace
asymptotics are no longer (local) minima of the potential $U(x)$, but
instead are only local minima of the potential constrained to the
submanifold $S$, so that $\nabla U(x_-) \parallel \hat m$ instead of
$\nabla U(x_-)=0$. The same is true for the effective saddle point
$x_s$, which is no longer a proper saddle of $U(x)$. Therefore, the
new basin of attraction, $\tilde B$, is now a subset of $S$ instead of
all of $\RR^n$, as is its boundary, $\partial \tilde B$. As before, we
write $M_s := M(x_s)$, $H_s = \nabla\nabla U(x_s)$ and $H_-=
\nabla\nabla U(x_-)$.

In fact, all arguments made in sections~\ref{sec:asympt-expans-bound}
and~\ref{sec:laplace-asymptotics} go through with the minor
modification of operating in the $(n-1)$-dimensional tangent spaces
$T_{x_-}S$ and $T_{x_s}S$ around the stable point and the saddle
instead of all of $\RR^n$. This is possible in
particular because at the saddle point $x_s$, the space of conserved
quantities $T_{x_s}S$ cannot be parallel to the separatrix $\partial
B$, or in other words $\hat n(x_s) \nparallel \hat m(x_s)$ (where we
recall that $M_s\hat m(x_s)=0$ and $H_s M_s \hat n(x_s) =
\mu^+$). In fact, as shown in lemma~\ref{lm:Hperp}, $\hat n$ and $\hat
m$ are perpendicular in the $H_s^{-1}$ inner product, which simplifies
the integration. Following similar arguments in section~\ref{sec:asympt-expans-bound} one can derived the mean first passage time with a conserved quantity to be
\begin{equation}
  \label{eq:wb-general-conserved}
  w_B(x_-) = \sqrt{\frac{\pi}{2\eps}} \frac{\int_{\tilde{B}} \rho_\infty(x)\,dx}{\int_{\partial \tilde{B}} \rho_\infty(u)\sqrt{\alpha(u)\beta(u)}\,du}\,.
\end{equation}
The volume and boundary integrals in (\ref{eq:wb-general-conserved}) can be evaluated using the invariant measure~(\ref{eq:gibbs}) and the Laplace approximation as in section~\ref{sec:laplace-asymptotics}. Note that additional correcting factors need to be introduced when we apply the Laplace
method and integrate 
over $\tilde B\subset S$ and $\partial \tilde B$, as shown by lemmas~\ref{lm:detp} and~\ref{lm:detpp} in the appendix. In particular, we get
\begin{align*}
  \int_{\tilde{B}}\rho_{\infty}(x)dx &= \int_{\tilde B} e^{-U(x)/\eps}\,dx \stackrel{\eps\to0}{=} e^{-U(x_-)/\eps} \int_{T_{x_-}\tilde B} e^{-\tfrac12 x\cdot H_- x}\,dx\\
  &= \sqrt{\frac{(2\pi\eps)^{n-1}}{\det H_-}}\left|\hat m \cdot H_-^{-1}\hat m\right|^{-1/2} e^{-U(x_-)/\eps} 
\end{align*}
at the stable fixed point, and
\begin{align*}
  \int_{\partial \tilde B} \sqrt{\alpha(u)\beta(u)} \rho_{\infty}(u)\,du &= \int_{\partial \tilde B} \sqrt{\alpha(u)\beta(u)} e^{-U(u)/\eps}\,du \\
  &\stackrel{\eps\to0}{=} \sqrt{\alpha(x_s)\beta(x_s)} e^{-U(x_s)/\eps} \int_{T_{x_s}\partial \tilde B} e^{-\tfrac12 u\cdot H_s u}\,du\\
  &= \sqrt{\frac{(2\pi\eps)^{n-2}}{\det H_-}}\frac{\sqrt{\alpha(x_s)\beta(x_s)}}{\left|\hat m \cdot H_s^{-1}\hat m\right|^{1/2}\left|\hat n \cdot H_s^{-1}\hat n\right|^{1/2}} e^{-U(x_s)/\eps} 
\end{align*}
at the saddle point. Recall $\alpha(x_s) = \hat n \cdot M_s\hat n$, $\beta(x_s) =\hat n\cdot M_s H_s\hat n$, and using lemma~\ref{lm:beta} and~\ref{lm:alpha} of the appendix, we arrive at our final result
\begin{equation}
  \label{eq:eyring-kramers-final}
  \tau = \frac{2\pi}{\mu_-} \sqrt{\frac{|\det H_s|}{\det H_-}} \sqrt{\frac{\hat m\cdot H_s^{-1}\hat m}{ \hat m\cdot H_-^{-1}\hat m}}\  e^{\Delta U/\eps},
\end{equation}
where we additionally used the fact that the expected time of
transitions $\tau$ is twice the expected time to exit, $w_B(x_-)$.

\begin{remark}
  If we define as $A\big|_V$ the restriction of an operator
  $A:\RR^n\to\RR^n$ to a subspace $V\subset\RR^n$, i.e.~$A\big|_V :
  V\to\RR^n$, then equation~(\ref{eq:eyring-kramers-final}) can be
  equivalently written via determinants of the Hessians restricted to
  the tangent spaces of the conserved manifold at the two relevant
  points,
  \begin{equation}
    \label{eq:eyting-kramers-final-restricted}
    \tau = \frac{2\pi}{\mu_-} \sqrt{\frac{\left|\det \left(H_s\big|_{T_{x_s}S}\right)\right|}{\det \left(H_-\big|_{T_{x_-}S}\right)}} \ e^{\Delta U/\eps}\,.
  \end{equation}
  While notationally more pleasing, this formulation is less readily
  implementable numerically, as it necessitates finding a basis for
  the tangent spaces and expressing the Hessians in this basis, while
  equation~(\ref{eq:eyring-kramers-final}) simply corrects for the
  single conserved quantity under knowledge of the vector $\hat m$.
\end{remark}

\begin{remark}
  Via a repeated application of our arguments (see
  lemma~\ref{lm:multiConserve}), one can generalize
  equation~(\ref{eq:eyring-kramers-final}) to multiple conserved
  quantities by choosing an appropriate basis for the space of conserved
  quantities. Concretely, for vectors $\{\hat m_1,\ldots,\hat m_k\}$
  normal to the conserved manifold $S$, under the assumption that
  $m_i$ are orthogonal in the $H^{-1}$ inner product, we obtain
  \begin{equation*}
    \tau = \frac{2\pi}{\mu_-} \sqrt{\frac{|\det H_s|}{\det H_-}} \sqrt{\frac{\hat m_1\cdot H_s^{-1}\hat m_1}{ \hat m_1\cdot H_-^{-1}\hat m_1}}\cdots \sqrt{\frac{\hat m_k\cdot H_s^{-1}\hat m_k}{ \hat m_k\cdot H_-^{-1}\hat m_k}}\  e^{\Delta U/\eps}\,.
  \end{equation*}
  The variant via restricted Hessians,
  equation~(\ref{eq:eyting-kramers-final-restricted}), remains
  unchanged in this case.
\end{remark}

\begin{remark}
  Of course it might be simpler, in particular in finite dimensional
  systems, to consider instead of the original stochastic evolution
  equation a reduced equation that eliminates variables to enforce the
  conservation constraint explicitly. For example, for a chemical reaction
  transforming molecule $A$ into $B$ and back, but with the total
  number $A+B$ conserved, one could instead consider stochastic
  dynamics in the difference $A-B$. While sometimes this approach is
  practical, and it must lead to identical results, it often produces
  complicated equations, in particular in the functional setting.
\end{remark}

\subsection{Functional gradient flows and stochastic hydrodynamics}
\label{sec:funct-grad-flows}


While the above discussion and derivation focuses on gradient flows
in $\RR^n$, following the work
in~\cite{jordan-kinderlehrer-otto:2006}, it has been realised that a
vast array of systems that originate from macroscopic limits of
microscopic interacting particle systems can similarly be interpreted as
gradient flows, on the space of probability measures, and as
generalised Wasserstein-gradient flows of an entropy functional. The
easiest example is the many-particle limit of non-interacting Brownian
walkers, in the large particle limit, $N\to\infty$, but interactions
with external forces, surrounding fluids, or inter-particle
interactions can be incorporated as well. For finite but large number
of particles, $N\gg 1$, one expects fluctuations of the order
$1/\sqrt{N}$ and arrives at a \emph{stochastic} evolution equation in
the form of a stochastic partial differential equation (SPDE),
generally summarised under the notion of \emph{fluctuating
  hydrodynamics}~\cite{bedeaux-mazur:1974, dean:1996,
  landau-lifshitz:2007}. If the underlying microscopic model is in
detailed balance, so is the resulting stochastic hydrodynamics
equation. For the example of $N=1/\eps$ non-interacting random walkers,
\begin{equation*}
  dX_i(t) = \sqrt{2D} dW_i(t)\,,
\end{equation*}
the limiting SPDE for the density $\rho(x,t)$ of walkers is formally
given by
\begin{equation*}
  \partial_t \rho(x,t) = D\Delta \rho(x,t) + \sqrt{2D\eps} \partial_x (\sqrt{\rho(x,t)} \eta(x,t))\,,
\end{equation*}
which is a functional gradient flow
\begin{equation*}
  \partial_t \rho = -M(\rho) \frac{\delta E[\rho]}{\delta \rho} + \sqrt{2\eps} M_{1/2}(\rho) \eta\,,
\end{equation*}
with
\begin{equation*}
  E[\rho] = \int \rho\log\rho \,dx\,,\quad\text{and}\quad M(\rho)\xi = D\nabla\cdot(\rho\nabla \xi)\,.
\end{equation*}
The above limiting equation is formal, and there is considerable
effort involved in making this intuition rigorous, in particular for
nonlinear equations and in higher dimensions. The precise mathematical
interpretation of the resulting SPDEs is subject to active
research~\cite{fehrman-gess:2023,
  djurdjevac-kremp-perkowski:2024}. This includes, but is not limited
to, the interpretation of the $\mathcal O(\eps)$-divergence term
in~(\ref{eq:gradient-flow-mobility}), which for many nonlinear
equations diverges and requires renormalization. For the purposes of
this paper, we retreat to the notion that ultimately, every numerical
computation relies on discretization and hence an ``UV'' cutoff that
regularises any divergences. In systems of physical meaning, such a
cut-off can be naturally justified as length scale where the continuum
limit breaks down, such as the size of a molecule for a fluid. In this
sense, any spatially continuous SPDE is to be interpreted as a
notational shorthand for a discrete system with an appropriate
physical cut-off length scale.

\subsection{Full computational scheme}
\label{sec:full-comp-scheme}

Given the above derivation, we now have a complete recipe for
computing mean first passage times for metastable stochastic
hydrodynamics. Concretely, in order to estimate the mean first passage
time out of a locally stable configuration, we apply the following steps:
\begin{enumerate}
\item Compute the saddle point (for example via edge tracking or
  gentlest ascent dynamics (GAD)~\cite{e-zhou:2011}) constrained to the
  submanifold restriction, respecting the conserved quantities.
\item Compute the Hessian around the effective saddle and stable fixed
  point by discretizing the continuous operator via some spatial
  discretization scheme.
\item Compute the spectrum of this Hessian, and correct for its action
  in conserved normal direction (the subspace perpendicular to mass
  conservation).
\end{enumerate}
The result will be a quantitative estimate for the mean first passage
time for the small noise limit. Notably, there is no fitting parameter
or additional assumption. The computation has to be done only once,
and can then be used for any noise strength $\eps$ (but of course will
be more accurate for smaller $\eps$).

In the following section, we will demonstrate the applicability of
this scheme to a number of examples, starting with a two-dimensional
and easy to visualise toy example in
section~\ref{sec:two-dimens-grad}, and then two stochastic partial
differential equations motivated from interacting particle systems and
stochastic hydrodynamics: the rupture time for liquid thin films in
section~\ref{sec:stoch-hydr-thin}, and the a socio-economic model of
urban separation in section~\ref{sec:soci-dynam-urban}.

\section{Two-dimensional gradient flow}
\label{sec:two-dimens-grad}

As a simple and easy to visualise example, we first consider a
double-well for $(x,y)\in\RR^2$ given by
\begin{equation}
  \label{eq:doublewell}
  U(x,y) = \tfrac14(1-x^2)^2  +\tfrac12y^2(x^2+\tfrac14)\,.
\end{equation}
While this potential has two minima, at $(-1,0)$ and $(1,0)$, and a
saddle at $(0,0)$, we want to modify the gradient flow with the
mobility matrix
\begin{equation*}
  M(x,y) = \tfrac12(1+x^2)\hat p \hat p^T,
\end{equation*}
for a normalised vector $\hat p\in\RR^2$. Since $M(x)$ has a zero
eigenvalue with corresponding eigenvector $\hat m = (\hat p)^\perp$,
the gradient flow
\begin{equation}
  \label{eq:2d-gradient}
  d(X_t,Y_t) = -M(X_t,Y_t)\nabla U(X_t,Y_t)\,dt + \eps \nabla\cdot M(X_t,Y_t) +   \sqrt{2\eps} M_{1/2}(X_t,Y_t)\,(dW_{x}, dW_{y})
\end{equation}
will always remain confined to the subspace
\begin{equation*}
  S = \{ (x,y)\in\RR^2 |\ (x\ y)\cdot \hat m = k\}\,.
\end{equation*}
In other words, the quantity $k = (x\ y)\cdot \hat m$ is a conserved
quantity of equation~(\ref{eq:2d-gradient}), similar to how mass is
conserved in fluctuating hydrodynamic equations. Its value throughout
remains that of the initial conditions of~(\ref{eq:2d-gradient}). This
situation is depicted in figure~\ref{fig:2d-example} (left).

\begin{figure}
  \begin{center}
    \includegraphics[width=0.45\textwidth]{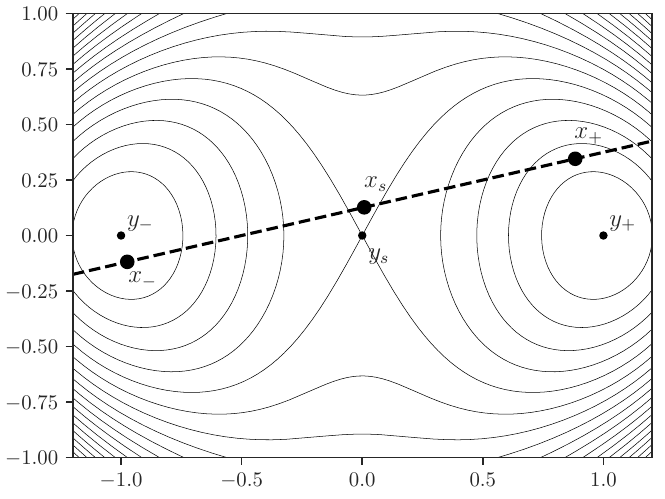}
    \includegraphics[width=0.45\textwidth]{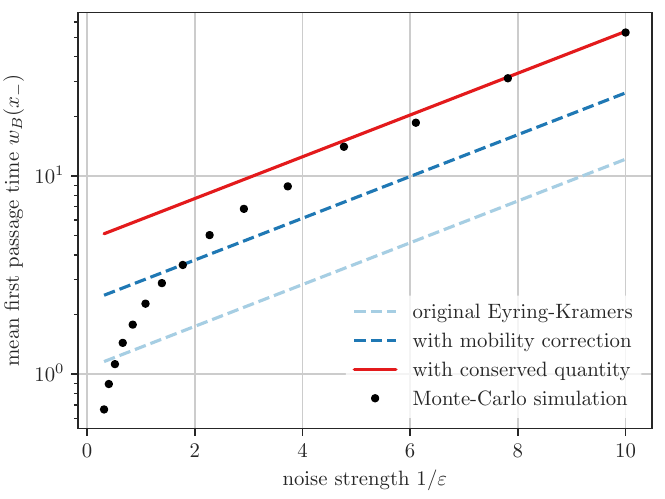}
  \end{center}
  \caption{\textbf{Left:} Doublewell~(\ref{eq:doublewell}) with conserved quantity. While
    the full double-well, with minima $y_\pm$ and saddle $y_s$, adheres
    to the original Eyring-Kramers
    formula~(\ref{eq:eyring-kramers-theirs}), the actual
    system~(\ref{eq:2d-gradient}) has a conserved quantity,
    restricting it to the dashed subspace. Not only does this result
    in different minima $x_\pm$ and saddle $x_s$ of the restricted
    system, but the ratio of Hessians in the Eyring-Kramers formula becomes
    incorrect, as it considers curvatures into suppressed
    directions. \textbf{Right:} Time to leave the basin of attraction
    of $x_-$ as a function of the noise amplitude $\eps$. Dots show
    the result of $1000$ numerical simulations of
    (\ref{eq:2d-gradient}) each, compared to the original
    Eyring-Kramers formula~(\ref{eq:eyring-kramers-nomobility}) (light
    blue dashed), the formula~(\ref{eq:eyring-kramers-theirs}) taking
    into account the mobility matrix (dark blue dashed), and finally
    our formula~(\ref{eq:eyring-kramers-final}) further taking into
    account the conserved quantities (red solid). Clearly, both
    corrections are needed to explain the observed
    times.\label{fig:2d-example}}
\end{figure}

For a numerical experiment to demonstrate the correctness of our
formula, we concretely pick $p=(1,\tfrac14)$ and $\hat p = p/|p|$, and
initialise with the conserved quantity set to $k=\tfrac18$. With these
values, we can numerically measure the mean time it takes to exit the
basin of attraction of the left well, and compare the results to our
formula in
section~\ref{sec:mean-first-passage}. The equation~\ref{eq:2d-gradient} is solved using a fourth order Runge-Kutta method~\cite{kasdin:1995} with timestep $dt=5\cdot 10^{-3}$. Figure~\ref{fig:2d-example}
(right) shows the results of the Monte-Carlo experiment, simulating
$N=1000$ samples for each value of $\eps$ and averaging the observed
time to exit the basin of attraction of the left well. This is
compared against the original Eyring-Kramers
formula~(\ref{eq:eyring-kramers-nomobility}) considering only the
properties of the Hessian (light blue dashed), the generalised
Eyring-Kramers formula including the correction from the mobility
operator, equation~(\ref{eq:eyring-kramers-theirs}) (dark blue
dashed), and lastly our final result~(\ref{eq:eyring-kramers-final})
which further considers the correction of the restriction to the
conserved subspace. As can be seen, the change of eigenvalue from
$\lambda_-$ to $\mu_-$, as well as the conserved quantity correcting
factor $\sqrt{(\hat m\cdot H^{-1}_s\hat m)/(\hat m\cdot H^{-1}_-\hat
  m)}$, both lead to a correcting factor of roughly 2, and both
corrections are needed in order to explain the observed result. We
stress that in order to obtain the fully corrected Eyring-Kramers
law~(\ref{eq:eyring-kramers-final}), a single computation needs to be
done for a prediction for all $\eps$, and without any fitting
parameter. The small noise limit, $\eps\ll 1$, appears to work
reasonably well already for values $\eps<\tfrac14$.

\section{Stochastic Hydrodynamics and Thin Film Rupture}
\label{sec:stoch-hydr-thin}

The stability of nanoscale thin liquid films on solid substrates plays
a key role in many applications including
coating~\cite{weinstein-ruschak:2004}, nanofluidic
transistors~\cite{karnik-fan-yue-etal:2005} and
nanomanufacturing~\cite{makarov-milichko-mukhin-etal:2016}. It has
been observed both
experimentally~\cite{herminghaus-jacobs-mecke-etal:1998,xie-karim-douglas-etal:1998,seemann-herminghaus-jacobs:2001}
and numerically with molecular dynamics
simulations~\cite{nguyen-fuentes-cabrera-fowlkes-etal:2014,zhang-sprittles-lockerby:2019,sprittles-liu-lockerby-etal:2023}
that initially flat films would rupture spontaneously, as shown in
figure~\ref{fig:STFE} (left). The classical explanation for the
rupture is due to the competition between the disjoining pressure
(otherwise known as the van der Waals forces) and the surface tension,
and a linear stability analysis~\cite{ruckenstein-jain:1974} further
reveals a critical wavelength above which the wave modes are linearly
unstable, eventually leading to rupture. However, subsequent
observations~\cite{seemann-herminghaus-jacobs:2001} have revealed a
larger set of regimes, one of which was hypothesised to stem from
\emph{thermally} activated rupture in the linearly stable regime.

\begin{figure}
  \begin{center}
    \includegraphics[width=0.45\textwidth]{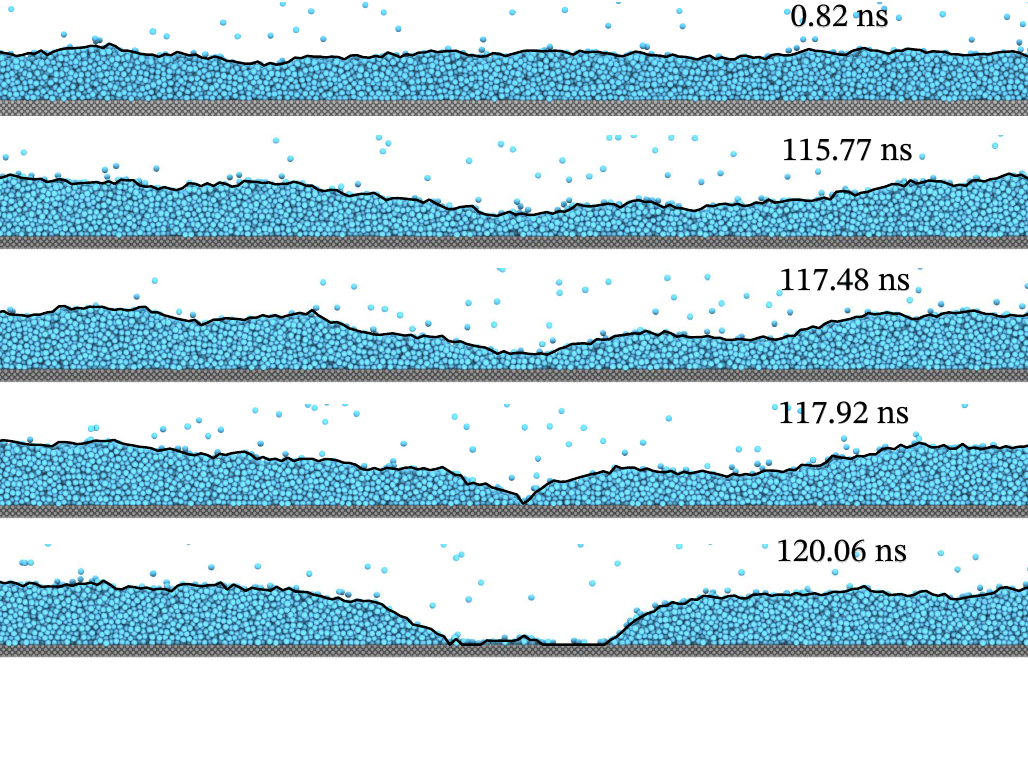}
    \includegraphics[width=0.45\textwidth]{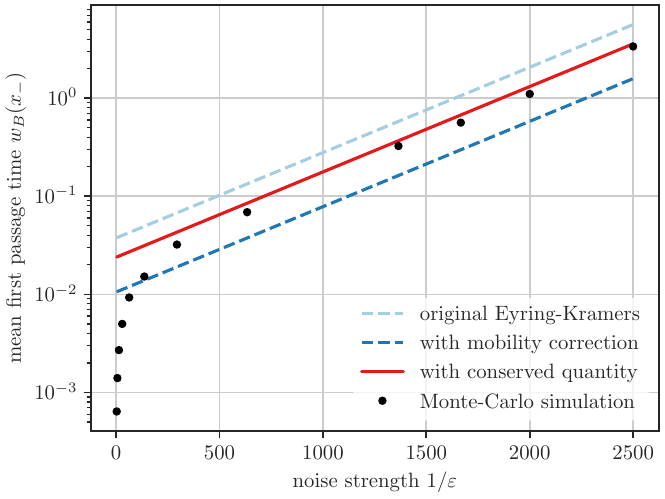}
  \end{center}
  \caption{\textbf{Left:} Snapshots from molecular dynamics simulation
    of a thin liquid film on a solid substrate. The blue particles
    indicate liquid and vapor. The silver particles indicate
    solid. The black lines show the position of the liquid-vapor
    interface, i.e. the film height. While the film is initially flat,
    is eventually ruptures by thermal fluctuations. \textbf{Right:} Average
    waiting time for rupture of thin film is plotted as a function of the inverse strength of the fluctuations. As $\varepsilon$ decreases, the observed average rupture time collapse onto the Eyring-Kramers prediction with our expression for the prefactor~(\ref{eq:eyring-kramers-final}).\label{fig:STFE}}
\end{figure}

Nanoscale films are difficult to observe experimentally and often molecular dynamics is utilised to explore their stability. However, MD can be computationally expensive, and thus there has been a drive towards developing macroscopic models to model this system. It has been
observed~\cite{gruen-mecke-rauscher:2006,duran-olivencia-gvalani-kalliadasis-etal:2019}
that the evolution of the thin film follows a stochastic hydrodynamic
limit, namely the stochastic thin film equation (STF) which, in the two-dimensional case, after
non-dimensionalisation~\cite{sprittles-liu-lockerby-etal:2023} reads
\begin{equation}
  \label{eq:stochastic-thin-film}
  \partial_t h(x,t) = \partial_x\left[c(h)\partial_x\left(-\partial^2_x h+\frac{4\pi^2}{3 h^3}\right)+\sqrt{2\varepsilon c(h)}\eta\right]\,.
\end{equation}
Here, $h(x,t)$ is the height of the thin film, $c(h)=h^3$ is the mobility associated with a no-slip solid, $\varepsilon$ is the noise amplitude and $\eta$ is a Gaussian white noise uncorrelated in both time and space, i.e. $\langle\eta(x,t)\eta(x',t')\rangle=\delta(x-x')\delta(t-t')$ where $\langle\;\rangle$ is the ensemble average and $\delta(x)$ is the Dirac delta functional. The STF is assumed to be periodic on $x\in[0,1]$ and the non-dimensionalisation is chosen so that the linear stability depends solely on the average film height $h_0=\int_0^1 h(x,t)dx=$ const: for $h_0>1$ the film is linearly stable and small perturbations without thermal fluctuations would decrease exponentially with time. For simplicity we use a constant mobility $c(h)=h_0^3$. The STF can also be interpreted as a functional gradient flow
\begin{equation*}
  \partial_t h(x,t) = -M(h)\frac{\delta E}{\delta h} + \sqrt{2\varepsilon} M_{1/2}(h)\eta\,,
\end{equation*}
for an energy functional
\begin{equation}
  \label{eq:STFE-energy}
  E[h] = \int_0^1 \left(\frac{1}{2}\left(\partial_x h\right)^2-\frac{2\pi^2}{3h^2}\right)dx\,,
\end{equation}
with mobility operator (acting on a test-function $\xi(x)$)
\begin{equation*}
  M(h)\xi = -\partial_x \left(h_0^3\partial_x\xi\right)\,,
\end{equation*}
and
\begin{equation*}
    M_{1/2}(h)\xi = \sqrt{h_0^3}\partial_x\xi\,.
\end{equation*}
The thermally activated rupture of the liquid nanofilm can then be
interpreted as a diffusive exit of the
SPDE~(\ref{eq:stochastic-thin-film}) from the basin of attraction of
the spatially constant solution
\begin{equation*}
  h(x,t) = h_0 >1\,.
\end{equation*}
Additionally, the system obeys mass
conservation, and as such constant functions are a zeromode of the
mobility operator. Therefore, in order to compute the expected time to
rupture, our full formalism~(\ref{eq:eyring-kramers-final}) is
necessary. Specifically, the computation consists of the following
steps: (1) We compute the saddle point of the energy
functional~(\ref{eq:STFE-energy}) via GAD, (2) we compute the second variation of the energy functional, acting on a test function $\xi(x)$ at an arbitrary point $h^*$, given by (see~\ref{app:hessian})
\begin{equation*}
    \frac{\delta^2 E[h(x)]}{\delta h(x)^2}\Big|_{h=h^*}\xi(x)= - \frac{4\pi^2}{h^*(x)^4}\xi(x)-\partial_{x}^2\xi(x)\,,
\end{equation*}
and compute numerically the spectrum of this operator at the fixed
point, $h^*=h_0$, and at the saddle, $h^*=h_s(x)$. The ratio of these
Hessians, evaluated according to
equation~(\ref{eq:eyring-kramers-nomobility}), yields the light blue
dashed line in figure~\ref{fig:STFE}. (3) Since the mobility operator
is not the identity, there is a correcting factor including $\mu_-$,
which we obtain numerically by computing the unique negative
eigenvalue of the operator
\begin{equation*}
  M(h_s(x)) \frac{\delta^2 E[h]}{\delta h^2}\Big|_{h=h_s}\xi = h_0^3\partial_x^2\left[\left(\frac{4\pi^2}{h_s(x)^4}+\partial_x^2\right)\xi\right]\,,
\end{equation*}
which yields instead the dark blue dashed line in figure~\ref{fig:STFE}. Lastly, we need to compute the action of the Hessian in direction of the vector normal to the conserved submanifold, which in this case is just the constant function $1(x)\equiv 1$. The inverse of the Hessian operator is evaluated numerically, and the result is the red solid line in figure~\ref{fig:STFE}, which agrees very well with the waiting time to rupture obtained via many stochastic Monte-Carlo experiments that integrate the stochastic thin film equation~(\ref{eq:stochastic-thin-film}) until a rupture is observed (black dots). The exponential time differencing method (ETD)~\cite{cox-matthews:2002} is used for the Monte-Carlo experiments with timestep $dt=1.566\cdot 10^{-7}$, and the details of implementation can be found in~\cite{sprittles-liu-lockerby-etal:2023}. Here we choose the average film height to be $h_0=1.01$, the STF is solved on a domain with $128$ uniformly distributed grid points, and the rupture times are averaged over $100$ events. Note that in \cite{sprittles-liu-lockerby-etal:2023}, additional molecular dynamics simulations demonstrated agreement of expected rupture times with equation~(\ref{eq:eyring-kramers-final}).

\begin{figure}
  \begin{center}
    \includegraphics[width=0.9\textwidth]{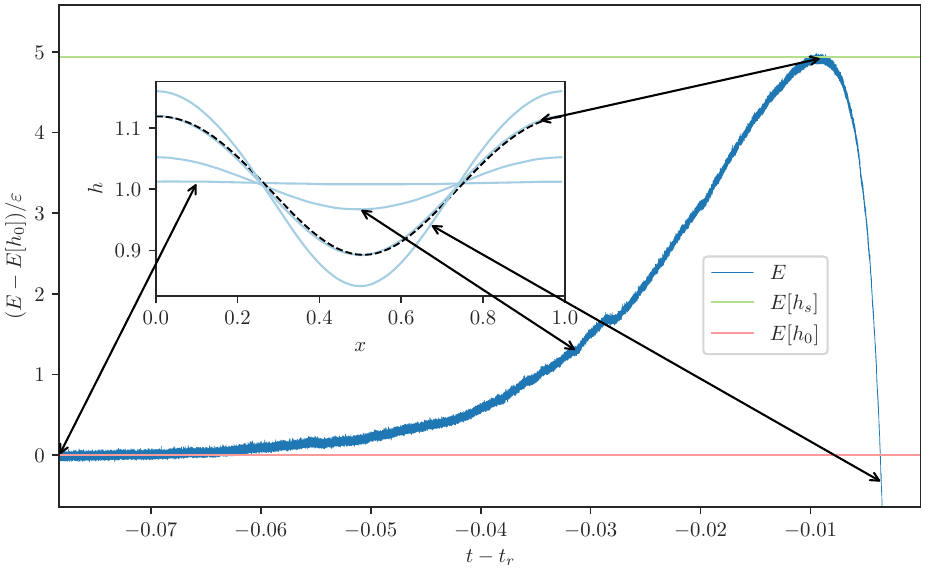}
  \end{center}
  \caption{Change of energy near rupture time. Noise amplitude is choosen to be $\varepsilon = 0.0005$. The dark blue line is the energy of the averaged profile, the light red line is the energy of the flat profile and the light green line is the energy of the saddle shape calculated analytically via GAD. The energy is translated by the energy of the flat profile and normalised by noise amplitude $\varepsilon$. The inset shows the averaged profile at different times with light blue lines and the analytical saddle shape with black dashed lines.\label{fig:STFE_energy}}
\end{figure}

To further characterise the rupture process, we investigate how the energy~(\ref{eq:STFE-energy}) changes with time near rupture. We perform $200$ independent simulations and record the film profiles for $5\cdot 10^5$ timesteps before the rupture time $t_r$. The energy is then calculated with the averaged profile to filter out the effect of thermal fluctuations, as shown in figure~\ref{fig:STFE_energy}. The light blue lines and the black dashed line in the inset show the averaged profiles at different times and the analytical saddle shape calculated from GAD, respectively. It is shown that the energy increase as the averaged profile deviates from its flat steady state, until the averaged profile reaches the saddle shape $h_s$ and drops dramatically. The analytical energy barrier is recovered from the simulations, and the analytical saddle shape agrees well with the averaged profile with maximum energy. These findings indicate that our saddle shape calculation is correct and the transition (or rupture) indeed goes through the saddle.

\section{Social dynamics and urban segregation}
\label{sec:soci-dynam-urban}

Fluctuating hydrodynamics SPDEs are not only encountered in continuum
limits of actual fluid models, but are regularly derived whenever
there is a large number of interacting agents, such as interacting
active particles~\cite{manacorda-puglisi:2017}, in traffic
flow~\cite{chu-yang-saigal-etal:2011}, pedestrian
dynamics~\cite{carrillo-martin-wolfram:2016,aurell-djehiche:2019} and
socioeconomic
interactions~\cite{zakine-garnier-brun-becharat-etal:2024}. In each
case, the particles are replaced by agents capable of acting according
to some simple ruleset. In socioeconomic models, a generic assumption
is that the agents try to individualistically improve their own
outcome or utility. In this context, the number of possible
equilibrium states of the overall model, as well as their relative
likelihood, becomes extremely important, as it describes directly the
most likely emergent state that the system will spontaneously converge
to. Consequently, the convergence to the ultimate stable state can be
seen as the manifestation of the ``invisible hand'' crystallizing the
collective societal state out of the individual agents' behavior.

A well-known example is the phenomenon of urban segregation, described
by the Schelling or Sakoda-Schelling models~\cite{schelling:1971,
  sakoda:1971}. In these, a large number of agents is prescribed, each
belonging to one of multiple distinct sub-populations, for example
representing social or ethnic background, which are free to relocate
depending on their preferences. In the original model, the presence of
only a slight preference of agents to surround themselves with
neighbors of their own sub-population led to completely segregated
geographical regions in the long-time limit. Following ideas
introduced in~\cite{grauwin-bertin-lemoy-etal:2009,
  burger-pietschmann-ranetbauer-etal:2022,
  zakine-garnier-brun-becharat-etal:2024}, these models can be
simplified to consist only of a single population, with spatial
exclusion and density dependent diffusivity, leading to a fluctuating
hydrodynamic equation of
\begin{equation}
  \label{eq:segregation}
  \partial_t \rho = \nabla\cdot((1-\rho)\nabla(D(\rho)\rho) + \rho D(\rho)\nabla \rho) + \nabla\cdot(\sqrt{\rho(1-\rho)} \eta(x,t))\,,
\end{equation}
where $\eta$ is spatio-temporally white noise and the diffusivity of
agents is given by
\begin{equation}
  \label{eq:diffusivity}
  D(\rho) = D_0 e^{-CK\star\rho}\,.
\end{equation}
This diffusivity exhibits a spatial convolution ``$\star$'' with a
kernel $K$, representing non-local sensing of their neighborhood by
each of the agents. In essence, equation~(\ref{eq:segregation})
describes the (nonlinear) diffusion of agents under spatial exclusion,
such that the density remains between 0 and 1, representing complete
absence of agents to full occupation. The density dependent
diffusivity~(\ref{eq:diffusivity}) represents the tendency of agents
to relocate towards a higher density of peers in the vicinity, up to
some maximum range given by a spatial cutoff of the (symmetric) kernel
$K(x,y)=K(x-y)$. This corresponds to energy and mobility given by
\begin{equation}
\label{eq:SEG_energy}
  \begin{cases}
    E[\rho] = \int\left(\rho\log \rho + (1-\rho)\log(1-\rho) - \tfrac12 C\rho K\star \rho\right)\,dx\\
    M(\rho) \xi = -\nabla\cdot(\rho(1-\rho) D(\rho) \nabla\xi)\,,\\
    M_{1/2}(\rho)\xi = \nabla\cdot(\sqrt{\rho(1-\rho)D(\rho)}\xi)\,,
  \end{cases}
\end{equation}
where the mobility again conserves total mass, and we are in
the framework where our results apply.

For simplicity, we assume the population density $\rho$ is periodic on domain $x\in[0,1]$. We also assume a Gaussian kernel $K(z) = \frac
1{\sqrt{2\pi}\kappa} \exp(-\frac12 z^2/\kappa^2)$ with sensing length
scale $\kappa>0$, corresponding to $\hat K(k) = \exp(-k^2\kappa^2/2)$
in Fourier space. Expanding the convolution for $\kappa\ll 1$ yields (see~\ref{app:expan_conv})
\begin{equation*}
\label{eq:SEG_conv}
  (K\star \rho)(x) = \rho(x) + \frac{\kappa^2}{2} \partial_x^2 \rho(x) + \mathcal O(\kappa^4)\,.
\end{equation*}

\begin{figure}
  \begin{center}
    \includegraphics[width=0.45\textwidth]{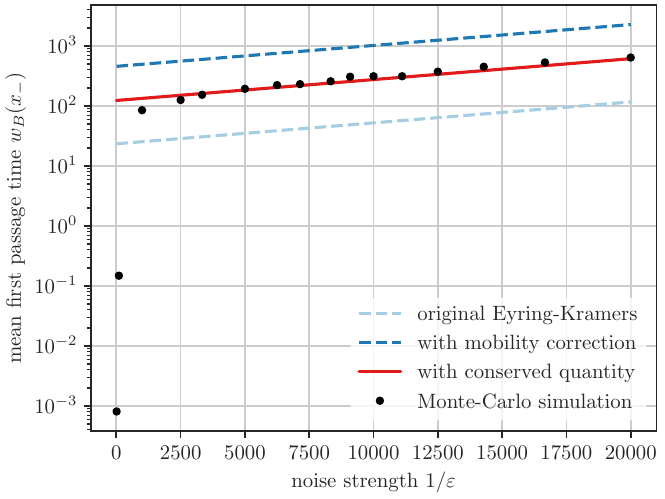}
    \includegraphics[width=0.45\textwidth]{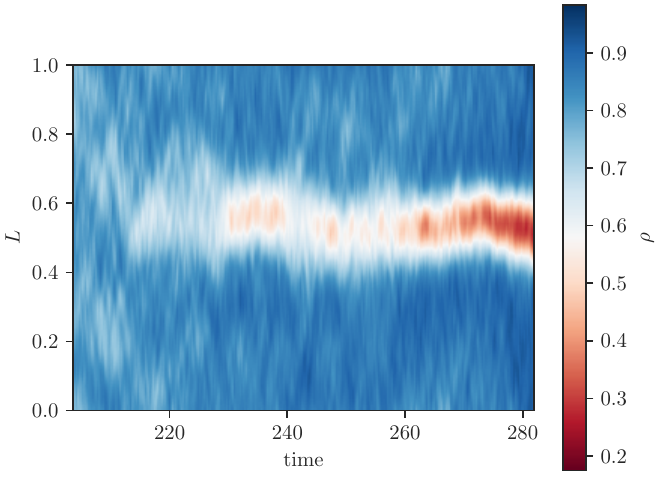}
  \end{center}
  \caption{\textbf{Left:} Average waiting time for the spontaneous
    segregation of the social dynamics model as a function of the
    inverse strength of fluctuations. After a transient for large
    noise, waiting times are exponentially distributed. The correct
    exponential distribution is correctly predicted, including its
    prefactor, by our formula~(\ref{eq:eyring-kramers-final}), while
    naive application of the Eyring-Kramers formula leads to
    mispredictions of a factor 5.\label{fig:SEG} \textbf{Right:}
    Evolution (time-space) of the density of agents near an observed
    segregation event. The originally homogeneous and fluctuating
    distribution of agents spontaneously segregates into a dense and
    a diluted region.}
\end{figure}
The Hessian operating on a periodic test function $\xi(x)$ can then be expressed by (see~\ref{app:hessian})
\begin{equation*}
    \frac{\delta^2E[\rho(x)]}{\delta\rho(x)^2}\xi(x) = \left(\frac{1}{\rho(x)}+\frac{1}{1-\rho(x)}-C\right)\xi(x)-\frac{\kappa^2}{2}C\partial^2_x \xi(x)\,.
\end{equation*}
For simplicity, we further assume that the population density is constant in the mobility operator, $\rho(x) = \bar{\rho}$, where $\bar\rho=\int_0^1\rho dx$ is the mass that is conserved. We can then calculate $\mu_-$ by numerically computing the unique negative eigenvalue of the operator
\begin{equation*}
    M(\bar\rho)\frac{\delta^2 E[\rho]}{\delta \rho^2}\xi = -\bar\rho(1-\bar\rho)D(\bar\rho)\partial^2_x\left[\left(\frac{1}{\rho}+\frac{1}{1-\rho}-C-\frac{\kappa^2}{2}C\partial^2_x\right)\xi\right]\,.
\end{equation*}

For the parameter $C=6, \kappa = 8\cdot10^{-3}$, $D_0 = 1$ and mass
$\bar\rho=\int \rho\,dx = 0.7908$, this system exhibits multiple stable
fixed points: A spatially homogeneous solution $\rho(x)=\bar \rho$, as
well as a localised (``cluster'' or ``aggregated'') state with minimum $\rho_{min} = 0.1747$, which
observes an aggregation of the agents in part of the domain, leaving
behind a depleted region where the concentration of agents is low. Both
of these states are locally stable and hence long lived for small
enough fluctuations. Fluctuations are understood to be spontaneous
local movements of agents, preserving their total number, but
rearranging them locally in space, with probabilities based both on
the local densities (through the exclusion terms proportional to
$\rho(1-\rho)$, which prevents movement out of empty or into fully
occupied regions), as well as their perceived relative attractiveness
encoded in the energy functional. As a result, an initially
homogeneous population will eventually spontaneously segregate due to
fluctuations. Population density $\rho$ is discretised with $64$ uniformly distributed grid points. Equation~(\ref{eq:segregation}) is solved numerically using the same exponential time difference scheme as in the previous section with timestep $dt=7.832\cdot 10^{-5}$. The population density is segregated when its minimum has reached $\rho_{min}$, and we record the waiting times averaged over $100$ realisations for different noise amplitude $\varepsilon$. These waiting times have an exponential distribution, correctly predicted by our formula~(\ref{eq:eyring-kramers-final}), as
shown in figure~\ref{fig:SEG} (left, red line), while using the
original Eyring-Kramers formula, or the mobility correction only leads
to mispredictions by a factor approximately 5.

Figure~\ref{fig:SEG} (right) shows a single segregation event, in
which an initially homogeneous population of agents is driven to
segregation by fluctuations that locally deplete the population strong
enough for a gap to form, transitioning into the segregated state with
a dilute region and an aggregate.

\section{Conclusion}
\label{sec:conclusion}

We show how recent breakthroughs in the derivation of mean first
passage times to leave the basin of attraction of metastable states
can be generalised to compute expected waiting times for wide classes
of (generalised) gradient flows in the presence of conserved
quantities. These generalizations are particularly important when
applied to fluctuating hydrodynamics equations, which are limiting
equations of interacting particle systems in the limit of many
particles.

Such equations are ubiquitous in nature, whenever a large number of
interacting agents leads to complex emergent behavior: Apart from
molecular dynamics and its applications in chemistry and material
design, systems such as traffic flows, pedestrian dynamics, or
socioeconomics must follow similar large-scale limits. All these
systems usually possess conserved quantities, such as mass or
number-of-agents, energy, or momentum, which lead to divergences in
the naive application of the limiting equations for mean first passage
times due to zero-eigenvectors in the corresponding mobility operator.

Here, we show how we can generalise existing results to incorporate
(1) position dependent mobility, (2) degenerate mobility operators
including zero-modes, (3) formally the functional setting, where we
are applying our results successfully to gradient flows in function
spaces. The result is a closed formula for the expected passage times
for leaving a locally stable basin of attraction of the stochastic
dynamics in the low-noise limit. We demonstrate our results to be
applicable in a broad class of settings, including liquid nanofilm
rupture times as well as social dynamics with urban segregation. The
results are very generally applicable to other systems of the same
class, including shallow water flows, Elo dynamics, or traffic flows.

\section*{Acknowledgments}

TG would like to thank A.~Donev for interesting discussions. TG
acknowledges the support received from the EPSRC projects EP/T011866/1
and EP/V013319/1. JES acknowledges the support from EPSRC grants
EP/W031426/1, EP/S029966/1 and EP/P031684/1. JBL was supported by a
studentship within the EPSRC–supported Centre for Doctoral Training in
modeling of Heterogeneous Systems, Grant No. EP/S022848/1. We
additionally want to thank the anonymous referee for suggesting the
generalization to multiple conserved quantities via restriction to the
tangent space. Relevant code used in this paper is openly available
at:
\url{https://github.com/JingBang-Liu/fluctuating_hydrodynamics_first_passage_time}.

\section*{References}

\bibliographystyle{iopart-num}
\bibliography{bib}

\appendix

\section{}

\begin{lemma}\label{lm:nEV}
  Let $s$ be the relevant saddle, and $\mu_-$ the unique unstable
  eigenvector of $M_s H_s$, and $\hat n$ the normal vector to
  $\partial B$ at the saddle. Then
  \begin{equation}
    \label{eq:mulem-st}
    H_s M_s\hat n = \mu_-\hat n\,,
  \end{equation}
  i.e. $\hat n$ is an eigenvector of $H_s M_s$ with eigenvalue $\mu_-$.
\end{lemma}
\begin{proof}
  Let $V^+=T_{x_s}\partial B$ be the tangent space to the separatrix
  at the saddle, which is spanned by the $n-1$ eigenvectors
  $\{v_i^+\}_{i\in\{1,\ldots,n-1\}}$ that correspond to positive 
  eigenvalues $\mu_i^+$ of $M_sH_s$. All these stable eigenvectors are
  parallel to the separatrix, implying $v_i^+\cdot \hat n=0$ for all
  $i$. Further, denote by $v^-$ the unique unstable eigenvector of
  $M_s H_s$ with eigenvalue $\mu_-$. Together, the $v_i^+$ and $v^-$
  span all of $\RR^n$ and we can write every vector $v\in\RR^n$ as
  \begin{equation}
    v = c^- v^- + \sum_i c_i^+ v_i^+\quad\text{with}\quad c^-,c_1^+,\ldots,c_{n-1}^+\in \RR\,.
  \end{equation}
  Then
  \begin{align*}
    \hat n\cdot M_sH_s v &= c^- \hat n\cdot M_sH_s v^- + \sum_i c_i^+ \hat n\cdot M_sH_s v_i^+\\
    &= \mu_- c^- \hat n\cdot v^- + \sum_i c_i^+ \mu_i^+ \underbrace{\hat n\cdot v_i^+}_{=0}\\
    &= \mu_- \left(c^- v^-\cdot \hat n + \sum_i c_i^+ \underbrace{v_i^+\cdot \hat n}_{=0}\right)
    &= \mu_- v\cdot \hat n\,.
  \end{align*}
  We conclude
  \begin{equation}
    \label{eq:mulem-1}
    v\cdot H_s M_s \hat n = \mu_- v\cdot \hat n\,.
  \end{equation}
  Since~(\ref{eq:mulem-1}) holds for arbitrary $v$, we obtain the
  statement~(\ref{eq:mulem-st}).
\end{proof}
\begin{lemma}\label{lm:beta}
  At the saddle, $x=x_s$,
  \begin{equation}
    \beta(x_s) = \hat n\cdot M_s H_s\hat n = \mu_-\,.
  \end{equation}
\end{lemma}
\begin{proof}
  Follows immediately from lemma~\ref{lm:nEV}.
\end{proof}

\begin{lemma}\label{lm:alpha}
  At the saddle, $x=x_s$,
  \begin{equation}
    \mu_- = \frac{\alpha(s)}{\hat n\cdot H_s^{-1}\hat n}\,.
  \end{equation}
\end{lemma}
\begin{proof}
  For $\alpha(s) = \hat n\cdot M_s\hat n$, we have
  \begin{equation}
    M_s\hat n = \mu_- H_s^{-1} \hat n
  \end{equation}
  from lemma~\ref{lm:nEV}. Solving for $\mu_-$ yields
  \begin{equation}
    \mu_- = \frac{\hat n\cdot M_s\hat n}{\hat n\cdot H_s^{-1}\hat n} = \frac{\alpha(s)}{\hat n\cdot H_s^{-1}\hat n}\,,
  \end{equation}
  which is the desired result.
\end{proof}

\begin{lemma}\label{lm:detp}
  Let $N$ be a co-dimension 1 hyperplane in $\RR^n$ with normal vector
  $\hat n$, and $H\in\RR^{n\times n}$ positive definite. Then, the
  Gaussian integral, restricted to the hyperplane $N$, is given by
  \begin{equation}
    \int_N e^{-\tfrac12 y\cdot Hy}\,d\sigma(y) = (2\pi)^{(n-1)/2} |\hat n\cdot H^{-1}\hat n|^{-1/2} |\det H|^{-1/2}\,.
  \end{equation}
\end{lemma}

\begin{proof}
  We have
  \begin{equation}
    \label{eq:detp-aux}
    \int_{\RR^n} e^{-\tfrac12 z\cdot Hz}\,dz = (2\pi)^{n/2} |\det H|^{-1/2}\,.
  \end{equation}
  In order to obtain a formula for the restricted Gaussian integral,
  consider the coordinate change
  \begin{equation}
    z = y + sH^{-1}\hat n\quad\text{with}\quad y\in N,\quad s\in\RR\,.
  \end{equation}
  Since we can write $H^{-1}\hat n=(\hat n\cdot H^{-1}\hat n)\hat n +
  v$, where $v\in N$, we know that the change of variables yields
  \begin{equation}
    dz = d(H^{-1}\hat n)\wedge dy = |\hat n\cdot H^{-1}\hat n|\,d\sigma(y)\,ds\,.
  \end{equation}
  Here, $d\sigma(z)$ is the differential element in the hyperplane $N$. Thus,
  \begin{align*}
    \int_{\RR^n} e^{-\tfrac12 z\cdot Hz}\,dz &= |\hat n\cdot H^{-1}\hat n|\int_\RR\int_N e^{-\tfrac12 (y+sH^{-1}\hat n)\cdot H(y+sH^{-1}\hat n)}\,d\sigma(y)\,ds\\
    &= |\hat n\cdot H^{-1}\hat n|\left(\int_N e^{-\tfrac12 y \cdot H y}\,d\sigma(y)\right)\left(\int_\RR e^{-\tfrac12s^2(\hat n\cdot H^{-1}\hat n)}\,ds\right)\\
    &= (2\pi)^{1/2} |\hat n\cdot H^{-1}\hat n|^{1/2} \int_N e^{-\tfrac12 y \cdot H y}\,d\sigma(y)\,,
  \end{align*}
and via equation~(\ref{eq:detp-aux}) we arrive at the desired result.
\end{proof}

\begin{lemma}\label{lm:detpp}
  Let $N$ and $M$ be two co-dimension 1 hyperplanes in $\RR^n$ with
  normal vectors $\hat n$ and $\hat m$, respectively, and $H\in\RR^{n\times n}$
  positive definite. Then, the Gaussian integral, restricted to the
  intersection of the two hyperplanes $N\cap M$, is given by
  \begin{equation}
    \int_{N\cap M} e^{-\tfrac12 y\cdot Hy}\,dy = (2\pi)^{(n-2)/2}
    |\hat n \cdot H^{-1}\hat n|^{1/2} |\hat m \cdot H^{-1}\hat m|^{1/2} (\det H)^{1/2}
  \end{equation}
  if $\hat n$ and $\hat m$ are orthogonal in the $H^{-1}$ inner product,
  \begin{equation}
    \hat n\cdot H^{-1}\hat m=0\,.
  \end{equation}
\end{lemma}
\begin{proof}
  With a similar argument as before, consider the coordinate change
  \begin{equation}
    z = y  + sH^{-1}\hat n + t H^{-1}\hat m\quad\text{with}\quad y\in N\cap M, \quad s,t\in \RR
  \end{equation}
  Then, the volume element yields
  \begin{equation}
    dz = dy\wedge d(H^{-1}\hat n)\wedge d(H^{-1}\hat m) = |(\hat n\cdot H^{-1}\hat n)(\hat m\cdot H^{-1}\hat m) - (\hat n\cdot H^{-1}\hat m)(\hat m\cdot H^{-1}\hat n)|\,d\sigma(y)\,ds\,dt\,.
  \end{equation}
  Since by assumption $\hat n\cdot H^{-1}\hat m = \hat m\cdot
  H^{-1}\hat n = 0$, we arrive at the desired result with the same in
  lemma~\ref{lm:detp}.
\end{proof}

\begin{lemma}\label{lm:Hperp}
  Let $\hat m$ be the zero eigenvector of the mobility matrix $M(x_s)$
  at the saddle $x_s$, and $\hat n$ the normal vector to the
  separatrix $\partial B$ at the saddle $x_s$. Then
  \begin{equation}
    \hat n\cdot H_s^{-1}\hat m = \hat m\cdot H_s^{-1}\hat n=0\,.
  \end{equation}
\end{lemma}
\begin{proof}
  From lemma~\ref{lm:nEV} we know
  \begin{equation}
    M_s \hat n = \mu_- H_s^{-1}\hat n\,,
  \end{equation}
  and thus
  \begin{equation}
    \hat m \cdot H_s^{-1} \hat n = \frac1{\mu_-} \hat m\cdot M\hat n = \frac1{\mu_-} \hat n\cdot M\hat m = 0\,.
  \end{equation}
\end{proof}

\begin{lemma}\label{lm:multiConserve}
    Let $N$ and $M_i$ be co-dimension 1 hyperplanes in $\RR^n$ with normal vectors $\hat n$ and $\hat{m}_i$, $i=1,\ldots,k$, respectively, and $H\in\RR^{n\times n}$ positive definite. Then, the Gaussian integral, restricted to the intersection of the hyperplanes $N\cap M_1\ldots\cap M_k$, is given by
    \begin{align}
        \int_{N\cap M_1\ldots\cap M_k}e^{-\frac{1}{2}x\cdot H x}dx &=  (2\pi)^{(n-k)/2}
        |\hat n \cdot H^{-1}\hat n|^{1/2}\\
        &\times|\hat{m}_1 \cdot H^{-1}\hat{m}_1|^{1/2}\ldots 
        |\hat{m}_k \cdot H^{-1}\hat{m}_k|^{1/2}(\det H)^{1/2},
    \end{align}
    if $\hat n$ and $\hat{m}_i$ are orthogonal in the $H^{-1}$ inner product.
\end{lemma}
\begin{proof}
    With a similar argument as before in lemma~\ref{lm:detpp}, consider the coordinate change
    \begin{equation*}
        z = y+sH^{-1}\hat n+\sum_{i=1}^kt_iH^{-1}\hat{m}_i\quad \text{with}\quad y\in N\cap M_1\ldots\cap M_k, \quad s,t_i\in \RR.
    \end{equation*}
    Then, by the assumption that $\hat n$ and $\hat{m}_i$ are orthogonal in the $H^{-1}$ inner product, we have the volumen element
    \begin{align*}
        dz &= dy\wedge d(H^{-1}\hat{n})\wedge d(H^{-1}\hat{m}_1)\ldots\wedge d(H^{-1}\hat{m}_k) \\
        &= |(\hat n\cdot H^{-1}\hat n)(\hat{m}_1\cdot H^{-1}\hat{m}_1)\ldots (\hat{m}_k\cdot H^{-1}\hat{m}_k)|d\sigma(y)\,ds\,dt_1\ldots dt_k,
    \end{align*}
    and we arrive at the desired result with the same in lemma~\ref{lm:detpp}.
\end{proof}

\section{Hessian as second variation}
\label{app:hessian}

In this section we sketch a formal derivation of Hessian of a given
energy functional. We first show the derivation of the Hessian of the
energy functional of the STF. The functional derivative of the energy
functional of STF~(\ref{eq:STFE-energy}) is given
by~\cite{parr-yang:1995} (here we omit in our notation the
$t$-dependence of $h$)
\begin{equation}
\label{eq:app_1st}
    \frac{\delta E[h(x)]}{\delta h(x)} = \frac{4\pi^2}{3h(x)^3} - \frac{\partial^2 h(x)}{\partial x^2}.
\end{equation}
The Hessian we are looking for is formally the functional derivative
of the functional derivative. If we rewrite
equation~(\ref{eq:app_1st}) as an integral,
\begin{align}
    \frac{\delta E[h(x)]}{\delta h(x)} &= F[h(x)] = \int_0^1 \left(\frac{4\pi^2}{3 h(x')^3}-\frac{\partial^2 h(x')}{\partial x'^2}\right)\delta(x'-x)dx' \\
    &= \int_0^1 f(x',h(x'),\frac{\partial^2 h(x')}{\partial x'^2})dx'\,,
\end{align}
we can again use the definition of functional derivative~\cite{parr-yang:1995} to get
\begin{align}
    \int_0^1\frac{\delta F[h(x)]}{\delta h(x)}\xi(x)dx &= \left\{\frac{d}{d\epsilon}(F[h(x)+\epsilon\xi(x)])\right\}_{\epsilon=0}\\
    &=\int_0^1\frac{\partial f}{\partial h}\xi(x') + \frac{\partial^2 f}{\partial (\partial^2_{x'}h(x'))^2}\frac{\partial^2\xi(x')}{\partial x'^2}dx'\\
    &=\int_0^1 \xi(x')\left(\frac{\partial f}{\partial h} + \frac{\partial^2}{\partial x'^2}\frac{\partial^2 f}{\partial (\partial^2_{x'}h(x'))^2}\right)dx'\\
    &=\int_0^1\xi(x')\left(-\frac{4\pi^2}{h(x')^4}\delta(x'-x)-\frac{\partial^2}{\partial x'^2}\delta(x'-x)\right)dx'\\
    &=-\frac{4\pi^2}{h(x)^4}\xi(x) - \frac{\partial^2\xi(x)}{\partial x^2}\\
    &=\int_0^1\frac{\delta^2 E[h(x)]}{\delta h(x)^2}\xi(x)dx\,.
\end{align}
Here $\xi(x)$ is a periodic test function, the third line used
integration by parts, and the fourth line used the properties of Dirac
delta functional and its derivatives. The Hessian, $\delta^2
E[h(x)]/\delta h(x)^2$, can be interpreted as an operator on $\xi(x)$,
and thus can be discretised and calculated numerically.

Similarly, we can calculate the gradient and the Hessian of the energy
functional of the urban segregation
model. Equations~(\ref{eq:SEG_energy}) and~(\ref{eq:SEG_conv}) give us
the following energy functional
\begin{equation}
    E[\rho] = \int_0^1\left(\rho\log \rho + (1-\rho)\log(1-\rho) - \tfrac12 C\rho^2 -\frac{\kappa^2}{4}C\rho \frac{\partial^2\rho}{\partial x^2}\right)dx\,.
\end{equation}
The functional derivative of $E[\rho]$ is
\begin{equation}
    \frac{\delta E[\rho]}{\delta \rho} = \log(\rho)-\log(1-\rho)-C\rho-\frac{\kappa^2}{2}C\frac{\partial^2\rho}{\partial x^2}\,,
\end{equation}
or in its integral form
\begin{align}
    \frac{\delta E[\rho]}{\delta \rho} &= F[\rho] = \int_0^1f(x',\rho(x'),\frac{\partial^2\rho(x')}{\partial x'^2}) \\
    &=\int_0^1\left(\log(\frac{\rho(x')}{1-\rho(x')})-C\rho(x')-\frac{\kappa^2}{2}C\frac{\partial^2\rho(x')}{\partial x'^2}\right)\delta(x'-x)dx'\,.
\end{align}
And the Hessian of $E[\rho]$ is given by
\begin{align}
    &\int_0^1\frac{\delta F[\rho(x)]}{\delta \rho(x)}\xi(x)dx = \int_0^1\frac{\partial f}{\partial \rho}\xi(x') + \frac{\partial^2 f}{\partial (\partial^2_{x'}\rho(x'))^2}\frac{\partial^2\xi(x')}{\partial x'^2}dx'\\
    =&\int_0^1\xi(x')\left((\frac{1}{\rho(x')}+\frac{1}{1-\rho(x')}-C)\delta(x'-x)-\frac{\kappa^2}{2}C\frac{\partial^2}{\partial x'^2}\delta(x'-x)\right)dx'\\
    =&\left(\frac{1}{\rho(x)}+\frac{1}{1-\rho(x)}-C\right)\xi(x)-\frac{\kappa^2}{2}C\frac{\partial^2 \xi(x)}{\partial x^2}\\
    =&\int_0^1\frac{\delta^2E[\rho(x)]}{\delta\rho(x)^2}\xi(x)dx.
\end{align}

\section{Local approximation of convolution}
\label{app:expan_conv}
In this section we show the local approximation of convolution with a
Gaussian kernel. Given a Gaussian kernel with variance $\kappa^2$,
\begin{equation}
  K(x) = \frac{1}{\sqrt{2\pi}\kappa}\exp\left(-\frac{x^2}{2\kappa^2}\right),
\end{equation}
we first show that its Fourier transform is also a Gaussian, that is
\begin{align}
  \hat{K}(k) = \int_{-\infty}^\infty \exp\left(-ikx\right)K(x)dx = \exp\left(-\frac{\kappa^2 k^2}{2}\right).
\end{align}
Differentiate the Gaussian kernel gives
\begin{equation}
  \frac{d K}{dx} = -\frac{x}{\kappa^2}K(x).
\end{equation}
Fourier transform on both side gives
\begin{equation}
  ik\hat{K}(k) = \frac{1}{i\kappa^2}\frac{d\hat{K}}{d k},
\end{equation}
and so
\begin{equation}
  \frac{1}{\hat{K}}\frac{d\hat{K}}{dk} = -k\kappa^2.
\end{equation}
Integrating both side from $0$ to $k$ gives
\begin{equation}
  \ln(\hat{K}(k))-\ln(\hat{K}(0)) = -\frac{k^2\kappa^2}{2}.
\end{equation}
Since the Gaussian kernel is normalised, we know that
\begin{equation}
  \hat{K}(0) = \int_{-\infty}^{\infty}K(x)\exp(0)dx = 1,
\end{equation}
and so we have the desired result. Assuming $\kappa^2 \ll 1$, we can then Taylor expand $\hat K = 1-\kappa^2 k^2/2+\mathcal O(\kappa^4)$, and so
\begin{equation}
  K(x)=\frac{1}{2\pi}\int_{-\infty}^{\infty}\hat K(k)\exp(ikx)dk =  \frac{1}{2\pi}\int_{-\infty}^{\infty}\left(1-\frac{\kappa^2}{2}k^2+\mathcal O(\kappa^4)\right)\exp(ikx)dk.
\end{equation} 
Since the Gaussian kernel is symmetric, the convolution is given by
\begin{align}
  K\star \rho &= \int_{-\infty}^{\infty} \rho(y)K(x-y) dy \\
  &=\int_{-\infty}^{\infty}\rho(y)\frac{1}{2\pi}\int_{-\infty}^{\infty}\left(1-\frac{\kappa^2}{2}k^2+\mathcal O(\kappa^4)\right)\exp(ik(x-y))dkdy\\
  &=\int_{-\infty}^{\infty}\frac{1}{2\pi}\hat \rho(k)\exp(ikx)dk - \frac{1}{2\pi}\int_{-\infty}^{\infty}\frac{\kappa^2}{2}k^2\hat \rho(k)\exp(ikx) dk + \mathcal O(\kappa^4)\\
  &=\rho(x)+\frac{\kappa^2}{2}\partial_x^2 \rho(x) + \mathcal O(\kappa^4).
\end{align}
\end{document}